\documentclass[a4paper,anonymous,USenglish,cleveref, autoref]{lipics-v2019}

\nolinenumbers

\title{Echo-CGC: A Communication-Efficient Byzantine-tolerant Distributed Machine Learning Algorithm in Single-Hop Radio Network}

\titlerunning{Echo-CGC: A Communication-Efficient Byzantine DML in Single-Hop Radio Network}

\author{Qinzi Zhang}{Boston College, USA}{zhangbcu@bc.edu}{}{}

\author{Lewis Tseng}{Boston College, USA}{lewis.tseng@bc.edu}{https://orcid.org/0000-0002-4717-4038}{}

\authorrunning{Q. Zhang and L. Tseng}

\Copyright{Qinzi Zhang and Lewis Tseng}



\category{}

\relatedversion{XXX}

\supplement{}

\acknowledgements{The authors would like to acknowledgement Nitin H. Vaidya and anonymous reviewers for their helpful comments..}

\EventEditors{John Q. Open and Joan R. Access}
\EventNoEds{2}
\EventLongTitle{OPODIS 2020}
\EventShortTitle{OPODIS 2020}
\EventAcronym{CVIT}
\EventYear{2016}
\EventDate{December 24--27, 2016}
\EventLocation{Little Whinging, United Kingdom}
\EventLogo{}
\SeriesVolume{42}
\ArticleNo{23}

\usepackage{tikz, pgfplots}
\usepackage{float}
\usepackage{algorithm}
\usepackage{algpseudocode}
\usepackage{caption, subcaption}

\algdef{SE}{Upon}{EndUpon}[1]{\textbf{upon}\ #1\ \algorithmicdo}{}
\algtext*{EndUpon}

\theoremstyle{plain}
\newtheorem{assumption}{Assumption}

\newcommand{\C}{\mathcal{C}}
\newcommand{\h}{\mathcal{H}}
\newcommand{\B}{\mathcal{B}}

\newcommand{\NN}{\mathbb{N}}

\newcommand{\RR}{\mathbb{R}}

\newcommand{\EE}{\mathbb{E}}

\DeclareMathOperator{\prob}{Pr}

\DeclareMathOperator*{\argmin}{argmin}

\newcommand{\norm}[1]{\lVert #1 \rVert}
\newcommand{\dotprod}[2]{\left< #1 , #2 \right>}

\begin{document}

\maketitle

\begin{abstract}


In the past few years, many Byzantine-tolerant distributed machine learning (DML) algorithms have been proposed in the point-to-point communication model. In this paper, we focus on a popular DML framework -- the parameter server computation paradigm and iterative learning algorithms that proceed in rounds, e.g., \cite{vaidya-2f-redundancy, Rachid_genuine_BML,Su_BGD}.  One limitation of prior algorithms in this domain is the \textit{high communication complexity}. All the Byzantine-tolerant DML algorithms that we are aware of need to send $n$ $d$-dimensional vectors from worker nodes to the parameter server in each round, where $n$ is the number of workers and $d$ is the number of dimensions of the feature space (which may be in the order of millions).
In a wireless network, power consumption is proportional to the number of bits transmitted.
Consequently, it is extremely difficult, if not impossible, to deploy these algorithms in power-limited wireless devices. 
Motivated by this observation, we aim to reduce the \textit{communication complexity} of Byzantine-tolerant DML algorithms in the \textit{single-hop radio network} \cite{Broadcast_SPAA10,Broadcast_Nitin_PODC05,Broadcast_Koo_PODC04}.

Inspired by the CGC filter developed by Gupta and Vaidya, PODC 2020 \cite{vaidya-2f-redundancy}, we propose a gradient descent-based algorithm, Echo-CGC. Our main novelty is a mechanism to utilize the \textit{broadcast properties} of the radio network to avoid transmitting the raw gradients (full $d$-dimensional vectors). In the radio network, each worker is able to overhear previous gradients that were transmitted to the parameter server.
Roughly speaking, in Echo-CGC,
if a worker ``agrees'' with a combination of prior gradients, it will broadcast the ``echo message'' instead of the its raw local gradient. The echo message contains a vector of coefficients (of size at most $n$) and the ratio of the magnitude between two gradients (a float). In comparison, the traditional approaches need to send $n$ local gradients in each round, where each gradient is typically a vector in a ultra-high dimensional space ($d \gg n$).
The improvement on communication complexity of our algorithm depends on multiple factors, including number of nodes, number of faulty workers in an execution, and the cost function. We numerically analyze the improvement, and show that with a large number of nodes, Echo-CGC reduces 80\% of the communication under standard assumptions.
\end{abstract}


\section{Introduction}

Machine learning has been widely adopted and explored recently \cite{ACM_Survey_DML_2020,ACM_Survey_DL_2020}. Due to the exponential growth of datasets and computation power required, distributed machine learning (DML) becomes a necessity. There is also an emerging trend \cite{IEEE_Survey_ML_Wireless,Wireless_FL_Chiang_2020} to apply DML in power-limited wireless networked systems, e.g., sensor networks, distributed robots, smart homes, and Industrial Internet-of-Things (IIoT), etc. In these applications, the devices are usually small and fragile, and susceptible to malicious attacks and/or malfunction. More importantly, it is necessary to reduce communication complexity so that (over-)communication does not drain the device battery. Most prior research on fault-tolerant DML (e.g., \cite{Rachid_genuine_BML,Rachid_Krum,vaidya-2f-redundancy,Su_BGD}) has focused on the use cases in clusters or datacenters. These algorithms achieve high resilience (number of faults tolerated), but also incur \textit{high communication complexity}. 
As a result, most prior Byzantine-tolerant DML algorithms are extremely difficult, if not impossible, to be deployed in power-limited wireless networks.

Motivated by our observations, we aim to design a Byzantine DML algorithm with reduced communication complexity. We consider wireless systems that are modeled as a \textit{single-hop radio network}, and focus on the popular parameter server computation paradigm (e.g., \cite{vaidya-2f-redundancy, Rachid_genuine_BML,Su_BGD}).
We propose \textit{Echo-CGC}, and prove its correctness under typical assumptions \cite{Rachid_Krum, Rachid_genuine_BML}. For the communication complexity, we formally analyze the expected number of bits that need to be sent from workers to the parameter server. 
The extension to multi-hop radio network is left as an interesting future work.

\textbf{Recent Development in Distributed Machine Learning}~~
Distributed Machine Learning (DML) is designed to handle a large amount of computation over big data. 
In the parameter server model, there is a centralized parameter server that distributes the computation tasks to $n$ workers. These workers have the access to the same dataset (that may be stored externally). Similar to \cite{Rachid_Krum,vaidya-2f-redundancy,Su_BGD}, we focus on the \textit{synchronous gradient descent} DML algorithms, where the server and workers proceed in synchronous rounds. In each round, each worker computes a local gradient over the parameter received from the server, and the server then aggregates the gradients collected from workers, and updates the parameter. Under suitable assumptions, prior algorithms \cite{Rachid_Krum,vaidya-2f-redundancy,Su_BGD} converge to the optimal point in the $d$-dimensional space $\RR^d$ even if up to $f$ workers may become Byzantine faulty.

To our knowledge, most Byzantine-tolerant DML or distributed optimization algorithms focused on the case of clusters and datacenters, which are modeled as a point-to-point network. For example, Reference \cite{Su_BGD}, Krum \cite{Rachid_Krum}, Kardam \cite{Rachid_Kardam}, and ByzSGD \cite{Rachid_genuine_BML} focused on the stochastic gradient descent algorithms under several different settings (synchronous, asynchronous, and distributed parameter server). Reference \cite{vaidya_nonBayesian_learning,vaidya-2f-redundancy,vaidya_optimization_PODC16} focused on the gradient descent algorithms for the general distributed optimization framework. Zeno \cite{Xie_Zeno} uses failure detection to improve the resilience. None of these works aimed to reduce communication complexity. 

Another closely related research direction is on reducing the communication complexity of non-Byzantine-tolerant DML algorithms, e.g., \cite{NIPS_Li_2014,Wireless_FL_Chiang_2020,HotEdge_Tao_2018}.
These algorithms are \textit{not} Byzantine fault-tolerant, and adopt a completely different design. For example, reference \cite{NIPS_Li_2014} utilizes relaxed consistency (of the underlying shared data), reference \cite{HotEdge_Tao_2018} discards coordinates (of the local gradients) aggressively, and reference \cite{Wireless_FL_Chiang_2020} uses intermediate aggregation. It is not clear how to integrate these techniques with Byzantine fault-tolerance, as these approaches reduce the redundancy, making it difficult to mask the impact from Byzantine workers.

\textbf{Single-Hop Radio Network}~~
We consider the problem in a single-hop radio network, which is a proper theoretical model for wireless networks. Following \cite{Broadcast_SPAA10,Broadcast_Nitin_PODC05,Broadcast_Koo_PODC04}, we assume that single-hop wireless communication is reliable and authenticated, and there is no jamming nor spoofing. Moreover, nodes follow a specific TDMA schedule so that there is no collision. 
In Section \ref{s:model}, we briefly argue why such an assumption is realistic to model wireless communication. In the single-hop radio network model, we aim to \underline{minimize the total number of bits} to be transmitted in each  round. If we directly adapt prior gradient descent-based algorithms \cite{Rachid_Krum,vaidya-2f-redundancy} to the radio network model, then each worker needs to broadcast a vector of size $d$, where $d$ is the number of dimensions of the feature space. In practical applications (e.g., \cite{Medical_ML,Wireless_FL_Chiang_2020}), $d$ might be in the order of millions, and the gradients may require a few GBs.  Since power consumption is proportional to the communication complexity in wireless channel, prior Byzantine DML algorithms  are \textit{not} adequate for power-limited wireless networks..

\textbf{Main Contributions}~~
Inspired by the CGC filter developed by Gupta and Vaidya, PODC 2020 \cite{vaidya-2f-redundancy}, we propose a gradient descent-based algorithm, \textit{Echo-CGC}, for the parameter server model in the single-hop radio network.
Our main observation is that since workers can overhear gradients transmitted earlier, they can use this information to avoid sending the raw gradients in some cases.
Particularly, 
if a worker ``agrees'' with some reference gradient(s) transmitted earlier in the same round, then they send a small message to ``echo'' with the reference gradient(s). The size of the echo message ($O(n)$ bits) is negligible compared to the raw gradient ($O(d)$ bits), since in typical ML applications, $d \gg n$.

Our proof is more sophisticated than the one in \cite{vaidya-2f-redundancy}, even though Echo-CGC is inspired by the CGC filter. The reason is that the ``echo message'' does \textit{not} necessarily contain worker $i$'s local gradient; instead, it can be used to construct an approximate gradient, which intuitively equals a combined gradients between $i$'s local gradients and the gradients broadcast by previous workers. We need to ensure that such an approximation does not affect the aggregation at the server. Moreover, CGC filter \cite{vaidya-2f-redundancy} works on deterministic gradients -- each worker computes the gradient of its local cost function using the full dataset. In our case, each worker computes a stochastic gradient, a gradient over a small random data batch. We prove that with appropriate assumptions, Echo-CGC converges to the optimal point. 

Echo-CGC is correct under the same set of assumptions in prior work \cite{Rachid_Krum};
however, there is an inherent trade-off between resilience, the proven bound on the communication complexity reduction, and the cost function. Fix the cost function. We derive necessary conditions on $n$ so that Echo-CGC is guaranteed to perform better. We also perform numerical analysis to understand the trade-off. In general, Echo-CGC saves more and more communication if $f/n$ becomes smaller and smaller. Moreover, our algorithm performs better when the variance of the data is relatively small. For example, our algorithm tolerates 10\% of faulty workers and saves over 75\% of communication cost when standard deviation of computed gradients is less than 10\% of the true gradient. 
\section{Preliminaries}

In this section, we formally define our models, and introduce the assumptions and notations.

\subsection{Models}
\label{s:model}

\textbf{Single-Hop Radio Network}~~
We consider the standard radio network model in the literature, e.g., \cite{Broadcast_SPAA10,Broadcast_Nitin_PODC05,Broadcast_Koo_PODC04}. In particular, the underlying communication layer ensures the \textit{reliable local  broadcast} property \cite{Broadcast_Nitin_PODC05}. In other words, the channel is perfectly reliable, and a local broadcast is correctly received by all neighbors.
As noted in \cite{Broadcast_SPAA10,Broadcast_Nitin_PODC05}, this assumption does not typically hold in the current deployed wireless networks, but it is possible to realize such a property with high probability in practice with the help from the MAC layer \cite{JAM_PODC08} or physical layer \cite{JAM_Infocom07}.

In our system, nodes can be uniquely identified, i.e., each node has a unique identifier. We assume that a faulty node may not spoof another node's identity. The communication network is assumed to be single-hop; that is, each pair of nodes are within the communication range of each other.
Moreover, time is divided into slots, and each node proceeds synchronously. Message collision is not possible because of the nodes follow a pre-determined TDMA schedule that determine the transmitting node in each slot and the transmission protocol is jam-resistant. Each slot is assumed to be large enough so that it is possible for a node to transmit a gradient. 
We also assume that each communication round (or communication step) is divided into $n$ slots, and the TDMA schedule assigns each node to a unique slot. For ease of discussion, node $i$ is scheduled to transmit at slot $i$.





\textbf{Stochastic Gradient Descent and Parameter Server}~~
In this work, we focus on the Byzantine-tolerant distributed Stochastic Gradient Descent (SGD) algorithms, which are popular in the optimization and machine learning literature \cite{Rachid_Krum,Rachid_genuine_BML,vaidya-2f-redundancy,convex-optimization}. Given a cost function $Q$, the (sequential) SGD algorithm outputs an optimal parameter $w^*$ such that
\begin{equation}
    \label{eq-objective}
    w^* = \argmin_{w\in\RR^d} Q(w)
\end{equation}
An SGD algorithm  executes in an iterative fashion, where in each round $t$, the algorithm computes the gradient of the cost function $Q$ at parameter $w^t$ and updates the parameter with the gradient.

\textit{Synchronous Parameter Server Model}~~
Computation of gradients is typically expensive and slow. One popular framework to speed up the computation is the \textit{parameter server model}, in which the parameter server distributes the computation tasks to $n$ workers and aggregates their computed gradients to update the parameter in each round. Following the convention, we will use node and worker interchangeably.

We assume a synchronous system, i.e., the computation and communication delays are bounded, and the server and workers know the bound. Consequently, if the server does not receive a message from worker $i$ by the end of some round, then the server identifies that worker $i$ is faulty.

Formally speaking, a distributed SGD algorithm in the parameter server model proceeds in synchronous rounds, and executes the following three steps in each round $t$:
\begin{enumerate}
    \item The parameter server broadcasts parameter $w^t$ to the workers.
    \item Each worker $j$ randomly chooses a random data batch $\xi_j^t$ from the dataset (shared by all the workers) and computes an estimate, $g_j^t$, of the gradient $\nabla Q(w^t)$ of the cost function $Q$ using $\xi_j^t$ and $w^t$.
    \item The server aggregates estimated gradients from all workers and updates the parameter using the gradient descent approach with step size $\eta$:
    \begin{equation}
    \label{eq:gd-uddate}
    w^{t+1}=w^t-\eta\sum_{j=1}^ng_j^t
    \end{equation}
\end{enumerate}

\textbf{Fault Model and Byzantine SGD}~~
Following \cite{vaidya-2f-redundancy,Rachid_Krum,Su_BGD}, our system consists of $n$ workers, up to $f$ of which might be Byzantine faulty. We assume that the central parameter server is always fault-free. 

Byzantine workers may be controlled by an omniscient adversary which has the knowledge of the current parameter (at the server) and the local gradient of all the other workers, and may have arbitrary behaviors. They \textit{can} send arbitrary messages. However, due to the reliable local broadcast property of the radio network model, they \textit{cannot} send inconsistent messages to the server and other workers. They also \textit{cannot} spoof another node's identity.
Our goal is therefore to design a distributed SGD algorithm that solves Equation (\ref{eq-objective}) in the presence of up to $f$ Byzantine workers.

Workers that are \textit{not} Byzantine faulty are called fault-free workers. These workers follow the algorithm specification faithfully. For a given execution of the algorithm, we denote $\h$ as the set of fault-free workers and $\B$ as the set of Byzantine workers. For brevity, we denote $h=|\h|$ and $b=|\B|$; hence, we have $b\leq f$ and $h\geq n-f$.

\textbf{Communication Complexity}~~
We are interested in minimizing the \underline{total number of bits} that need to be transmitted from workers to the parameter server in \textit{each round}. Prior algorithms \cite{vaidya-2f-redundancy,Rachid_Krum} transmit $n$ gradients in a $d$-dimensional space in each round, since each node needs to transmit its local gradient to the centralized server. Typically, each gradient consists of $d$ floats or doubles (i.e., a single primitive floating point data structure for each dimension). 

\subsection{Assumptions and Notations}
\label{sec:assumptions}

We assume that the cost function $Q$ satisfies some standard properties used in the literature \cite{Rachid_Krum,Rachid_genuine_BML,Su_BGD}, including convexity, differentiability, Lipschitz smoothness, and  strong convexity.
Following the convention, we use $\dotprod{a}{b}$ to represent the dot product of two vectors $a$ and $b$ in the $d$-dimensional space $\RR^d$.
\begin{assumption}[Convexity and smoothness]
    \label{as:convex-diff}
    $Q$ is convex and differentiable.
\end{assumption}
\begin{assumption}[$L$-Lipschitz smoothness]
    \label{as:lipschitz}
    There exists $L>0$ such that
    for all $w,w'\in\RR^d$,
    \begin{equation}
        \norm{\nabla Q(w)-\nabla Q(w')}\leq L\norm{w-w'}
    \end{equation}
\end{assumption}
\begin{assumption}[$\mu$-strong convexity]
    \label{as:strong-convex}
    There exists $\mu>0$ such that for all $w,w'\in\RR^d$,
    \begin{equation}
        \dotprod{\nabla Q(w)-\nabla Q(w')}{w-w'}\geq \mu\norm{w-w'}^2
    \end{equation}
\end{assumption}

We also assume that the random data batches are independently and identically distributed from the dataset. Before stating the assumptions, we formally introduce the concept of randomness in the framework.
Similar to typical stochastic gradient descent algorithms, the only randomness  is due to the random data batches $\xi_j^t$ sampled by each fault-free worker $j\in\h$ in each round $t$, which further makes $g_j^t$ as well as $w^{t+1}$ non-deterministic. In the case when a worker uses the entire dataset to train model, $g_j^t=\nabla Q(w^t)$. Hence, the result is deterministic, i.e., each fault-free worker derives the same gradient. In practice, data batch is a small sample of the entire data set.\footnote{Reference \cite{vaidya-2f-redundancy} works on a different formulation in which each worker may have a different local cost function.}

Formally speaking, we denote an operator $\EE_{\Xi^t}(\cdot\mid w^t, \mathcal{G}_\B^t)$ as the \textit{conditional expectation} operator over the set of random batches $\Xi^t=\{\xi_j^t,j=1,2,\ldots,n\}$ in round $t$ given (i) the parameter $w^t$, and (ii) the set of Byzantine gradients $\mathcal{G}_\B^t=\{g_j^t:j\in\B\}$. This conditional expectation operator allows us to treat $w^t$, $Q(w^t)$, and $\nabla Q(w^t)$ as constants, as well as the Byzantine gradients. This is reasonable because (i) we have the knowledge about $Q$ and $w^t$ given an execution, and (ii) the Byzantine gradients are arbitrary, and do not depend on the data batches. From now on, without further specification, we abbreviate the operator $\EE_{\Xi^t}(\cdot\mid w^t, \mathcal{G}_\B^t)$ as $\EE$.

Below we present two further assumptions of local stochastic gradient $g^t_j$ at each fault-free worker $j$. Similar to \cite{Rachid_Krum, Rachid_genuine_BML}, we rely on the two following assumptions for correctness proof. 

\begin{assumption}[IID Random Batches]
    \label{as:iid}
    For all $j\in\h$ and $t \in \NN$,
    \begin{equation}
        \EE(g_j^t)=\nabla Q(w^t)
    \end{equation}
\end{assumption}
\begin{assumption}[Bounded Variance]
    \label{as:bounded-variance}
    For all $j\in\h$ and $t \in \NN$,
    \begin{equation}
        \EE\norm{g_j^t-\nabla Q(w^t)}^2\leq \sigma^2\norm{\nabla Q(w^t)}^2
    \end{equation}
\end{assumption}

\textbf{Notation}~~
We list the most important notations and constants used in our algorithm and analysis in the following table.

\begin{table}[H]
    \centering
    \begin{tabular}{| l | l |}
        \hline 
        $\h$ & set of fault-free workers; $h=|\h|$ \\
        \hline 
        $\B$ & set of faulty workers; $b=|\B|$ \\
        \hline 
        $t$ & round number, $t=0,1,2,\ldots$ \\
        \hline
        $w^*$ & optimal solution to $Q$, i.e., $w^*=\argmin_{w\in\RR^d} Q(w)$ \\
        \hline 
        $w^t$ & parameter in round $t$ \\
        \hline 
        $g_j^t$ & estimated gradient of $j$ in round $t$ \\
        \hline 
        $\Tilde{g}_j^t$ & “reconstructed" gradient of $j$ by server in round \\
        \hline 
        $\hat{g}_j^t$ & gradient of $j$ in round $t$ after applying the CGC filter \\
        \hline
        $\eta$ & fixed step size as in Equation (\ref{eq:gd-uddate}) \\
        \hline
        $L$ & Lipschitz constant \\
        \hline 
        $\mu$ & strong convexity constant \\
        \hline 
        $r$ & deviation ratio, a key parameter in our algorithm \\
        \hline 
        $k^*$ & constant defined in Lemma \ref{lem:kx}, $k^*\approx1.12$ \\
        \hline
    \end{tabular}
    \caption{Notations and constants used in this paper.}
    \label{tab:notation}
\end{table}
\section{Our Algorithm: Echo-CGC}

Our algorithm is inspired by Gupta and Vaidya \cite{vaidya-2f-redundancy}. Specifically, we  integrate their CGC filter with a novel aggregation phase.
Our aggregation mechanism utilizes the broadcast property of the radio network to improve the communication complexity. In the CGC algorithm  \cite{vaidya-2f-redundancy}, each worker needs to send a $d$-dimensional gradient to the server, whereas in our algorithm, some workers only need to send the ``echo message'' which is of size $O(n)$ bits. Note that in typical machine learning applications, $d \gg n$.

We design our algorithm for the synchronous parameter server model, so the algorithm is presented in an iterative fashion. That is, each worker and the parameter server proceed in synchronous rounds, and the algorithm specifies the exact steps for each round $t$. Our Algorithm, Echo-CGC, is presented in Algorithm \ref{alg:echo-CGC}.  The algorithm uses the notations and constants summarized in Table \ref{tab:notation}.

\paragraph*{Algorithm Description}

Initially, the parameter server randomly generates an initial parameter $w^0\in\RR^d$. Each round $t \geq 0$ consists of three phases: 
(i) computation phase, (ii) communication phase, and (iii) aggregation phase. Echo-CGC takes the following inputs: step size $\eta$, deviation ratio $r$, number of workers $n$, and maximum number of tolerable faults $f$. The exact requirements on the values of these inputs will become clear later. For example, $n, f, r$ need to satisfy the bound derived in  Lemma \ref{lem:conv-beta}. More discussion will be presented in Section \ref{s:communication}.

\textbf{Computation Phase}~~
In the computation phase of round $t$, the server broadcasts $w^t$ to the workers. Each worker $j$ then computes the local stochastic gradient $g_j^t=\nabla Q_j(w^t)$ using $w^t$ and its random data batch $\xi_j^t$. Since we assume the parameter server is fault-free, each worker receives the identical $w^t$. The local gradient is stochastic, because each worker uses a random data batch to compute the local gradient $g^t_j$.

\textbf{Communication Phase}~~
In the communication phase, each worker needs to send the information regarding to its local gradient to the parameter server. This phase is our main novelty, and different from prior algorithms \cite{vaidya-2f-redundancy,Rachid_Krum,Su_BGD}. We utilize the property of the broadcast channel to reduce the communication complexity. 
As mentioned earlier, the communication phase of round $t$ is divided into $n$ slots $t_1,\ldots,t_n$. Without loss of generality, we assume that each worker $j$ is scheduled to broadcast its information in slot $t_j$ (of round $t$). Note that we assume that the underlying physical or MAC layer is jamming-resistant and reliable; hence, each fault-free worker can reliably broadcast the information to all the other nodes.

\textit{Steps for Worker $j$}:~~
Each worker $j$ stores a set of gradients that it overhears in round $t$. Denote by $R_j$ the set of stored gradients. By assumption, $R_j$ consists of gradients $g_i^t$ for $i < j$, when at the beginning of slot $t_j$. Upon receiving a gradient $g_i^t$ (in the form of a vector in $\RR^d$), worker $j$ stores it to $R_j$ if $g_i^t$ is linearly independent with all existing gradients in $R_j$. 
In the slot $t_j$, worker $j$ computes the ``echo gradient'' 
using vectors stored in $R_j$. Specifically, worker $j$ takes the following steps: 

\begin{itemize}
    \item It expresses $R_j$ as $R_j=\{g_{i_1}^t,\ldots,g_{i_{|R_j|}}^t\}$ and constructs a matrix $A_j\in\RR^{d\times|R_j|}$ as
\begin{equation}
    A_j^t = \begin{bmatrix}
        g_{i_1}^t & g_{i_2}^t & \dotsm & g_{i_{|R_j|}}^t
    \end{bmatrix} \notag
\end{equation}

\item It then computes the Moore-Penrose inverse (M-P inverse in short) of $A_j^t$, defined as
\begin{equation}
    (A_j^t)^+ = ((A^t_j)^TA^t_j)^{-1}(A^t_j)^T, \notag 
\end{equation}
where $A^T$ is the transpose of matrix $A$. The existence of the M-P inverse is guaranteed. Intuitively this is because all columns of $A_j^t$ are linearly independent by construction. The formal proof is presented in Appendix \ref{appendix-MP-inverse}.

\item Next, worker $j$ computes a vector $x_j^t\in\RR^{|R_j|}$ using the M-P inverse:
\begin{equation}
    x_j^t = (A_j^t)^+ g_j^t, \notag
\end{equation}
where $g_j^t$ is the local stochastic gradient of $Q$ computed by $j$ in the computation phase. 
Note that $x^t_j$ is of size $O(n)$, since $R_j$ contains at most $n$ elements.

\item Finally, it computes the ``echo gradient'' as 
\begin{equation}
(g_j^t)^*=A_j^tx_j^t \notag
\end{equation}
Mathematically, $(g_j^t)^*$ is the projection of $g_j^t$ onto the span of vectors in $R_j$, i.e., the closest vector to $g^t_j$ in the span of $R_j$.
 
\end{itemize}

Next, worker $j$ checks whether the following inequality holds where $(g_j^{t})^*$ is the echo gradient, $g_j^t$ the local stochastic gradient, and $r$ the deviation ratio.
\begin{equation}
    \label{eq-sendcheck}
    \norm{(g_j^{t})^*-g_j^t} \leq r\norm{g_j^t}
\end{equation}
Worker $j$ performs one of the two actions depending on the result of Inequality (\ref{eq-sendcheck}).
\begin{itemize}
    \item If Inequality (\ref{eq-sendcheck}) holds, then $j$ sends the \textit{echo message} $(\norm{g_j^t}/(\norm{g_j^{t})^*},x_j^t,I_j^t)$ to the server, where $I_j^t=\{i_1,\ldots,i_{|R_j|}\}$ is 
    a sorted list of worker IDs whose gradients are stored in $R_j$. 
    \item Otherwise, worker $j$ broadcasts the raw gradient $g_j^t$ to server and all the other workers. 
\end{itemize}

\textit{Steps for Parameter Server}:~~
The parameter server uses a vector $G$ to store the  gradients from workers. Specifically, in each round $t$, for each worker $j$, the server computes $\tilde{g}_j^t$ and stores it as the $j$-th element of $G$. 
At the beginning of round $t$, every element $G[j]$ is initialized as an empty placeholder $\perp$.
During the communication phase, the parameter server takes two possible actions upon receiving a message from worker $j$: 

\begin{itemize}
    \item If the message is a vector, then the server stores $\Tilde{g}_j^t=g_j^t$ in $G[j]$.
    \item Otherwise, the message is a tuple $(k,x,I)$. The server then does the following:
    \begin{itemize}
        \item If there exists some $i\in I$ such that $G[i]=\perp$ (i.e., the server has not received a message from worker $i$), then due to the reliable broadcast property, the server can safely identify $j$ as a Byzantine worker. By convention, we let the server store $\Tilde{g}_j^t=\Vec{0}$, the zero vector in $\RR^d$, in $G[j]$. 
        \item Otherwise, denote the matrix $A_I$ as $A_I=
        \begin{bmatrix}
            G[i_1], & \ldots, & G[i_{|R_j|}]
        \end{bmatrix}$ where $I=\{i_1,\ldots,i_{|R_j|}\}$,
        and the server stores $\Tilde{g}_j^t$ as $\Tilde{g}_j^t = kA_{I}x$ in $G[j]$. 
    \end{itemize}
\end{itemize}

\textbf{Aggregation Phase}~~
The final phase is identical to the algorithm in \cite{vaidya-2f-redundancy}, in which the server updates the parameter using the CGC filter. First, the server sorts the stored gradients $G^t$ in the increasing order of their Euclidean norm and relabel the IDs so that 
$\norm{\Tilde{g}_{i_1}^t}\leq \dotsm \leq \norm{\Tilde{g}_{i_n}^t}$. Then the server applies the CGC filter as follows:
\begin{equation}
    \label{eq-CGC}
    \hat{g}_j^t = \left\{
    \begin{array}{ll}
        \frac{\norm{\Tilde{g}_{i_{n-f}}^t}}{\norm{\Tilde{g}_j^t}}\Tilde{g}_j^t, & j\in\{i_{n-f+1},\ldots,i_n\} \\
        \Tilde{g}_j^t, & j\in\{i_1,\ldots,i_{n-f}\} \\
    \end{array}\right.
\end{equation}
Finally, the server aggregates the gradients by $g^t=\sum_{j=1}^n\hat{g}_j^t$ and updates the parameter by $w^{t+1}=w^t-\eta g^t$, where $\eta$ is the fixed step size. 

\begin{algorithm}
\caption{Algorithm Echo-CGC}
\label{alg:echo-CGC}
\begin{algorithmic}[1]
    \State \textbf{Parameters}: 
    \State \hspace{0.2in}$\eta>0$ is the step size defined in Equation (\ref{eq:gd-uddate}) 
    \State \hspace{0.2in}$r>0$ is the deviation ratio
    \State \hspace{0.2in}$n,f,r$ satisfy the resilience bounds stated in Lemma \ref{lem:conv-beta}
    \State \textbf{Initialization at server}:~~ $w^0\gets$ a random vector in $\RR^d$
    \For {$t\gets 0$ to $\infty$}
        \State /* Computation Phase */
            \State At server: ~~broadcast $w^t$ to all workers; ~~$G\gets$ a $\perp$-vector of length $n$
            \State At worker $j$: 
            \State \hspace{0.2in} receive $w^t$ from the server
            \State \hspace{0.2in}  $g_j^t\gets \nabla Q_j(w^t)$; ~~$R_j\gets\{\}$ \Comment{local stochastic gradient at worker $j$}
            \vspace{5pt}
        \State  /* Communication Phase */
            \For {$i\gets 1$ to $n$}
                \State  (i) At worker $i$: 
                
                \If {$|R_i|=0$}
                    \State broadcast $g_i^t$
                \Else
                    \State $A\gets [g]_{g\in R_j}$;~~ $A^+\gets (A^TA)^{-1}A^T$;~~ $x\gets A^+g_i^t$ \Comment{$Ax$ is the echo gradient}
                    \If {$\norm{Ax-g_i^t}\leq r\norm{g_i^t}$}
                        \State $I\gets\{i':{g_{i'}^t\in R_j}\}$ in an ascending order
                        \State broadcast $(\norm{g_i^t}/\norm{Ax},x,I)$ \Comment{echo message}
                    \Else 
                        \State broadcast $g_i^t$ \Comment{raw local gradient}
                    \EndIf 
                \EndIf
                \State  (ii) At worker $j>i$: 
                \If{$j$ receives vector $g_i^t$ from worker $i$}
                    \State $A\gets [g]_{g\in R_j}$;~~ $A^+\gets (A^TA)^{-1}A^T$
                    \If {$g_i^t$ is linearly independent with $R_j$ (i.e., $AA^+g_i^t\neq g_i^t$)}
                    \label{line:linear-independent}
                        \State $R_i\gets R_i\cup\{g_i^t\}$
                    \EndIf
                \EndIf    
                \Statex \hspace{2.7em} (iii) At server:
                \If {it receives a vector $g_j^t$ from worker $j$}
                    
                    \State $G[j] \gets g_j^t$
                    \Comment{$j$ transmitted a raw gradient} 
               \ElsIf {it receives an echo message $(k,x,I)$ from worker $j$}
                    \If {$\exists i\in I$ such that 
                $G[i] = \perp$    }
                        
                        \State $G[j]\gets \Vec{0}$
                        \Comment{$j$ is a Byzantine worker}
                    \Else 
                        \State $A_I\gets [\tilde{g}_i^t]_{i\in I}$,~~ $G[j]\gets kA_Ix$
                        \Comment{$j$ transmitted an echo message}
                    
                    \EndIf
                    
                \EndIf 
            \EndFor 
        \State /* Aggregation Phase (applying CGC filter from \cite{vaidya-2f-redundancy}) */
        \State $g^t\gets \sum_{g\in G} CGC(g)$ \Comment{$CGC(\cdot)$ defined in Equation (\ref{eq-CGC})}
        \State $w^{t+1}\gets w^t-\eta\cdot g^t$ \Comment{$\eta$ defined in Equation (\ref{eq:gd-uddate})}
    \EndFor 
\end{algorithmic}
\end{algorithm}

\section{Convergence Analysis}

In this section, we prove the convergence of our algorithm Echo-CGC. The proof is more complicated than the one in \cite{vaidya-2f-redundancy}, even though both algorithms use the CGC filter. This is mainly due to two reasons: (i) we use stochastic gradient, whereas  \cite{vaidya-2f-redundancy} uses a deterministic gradient; and (ii) echo messages only results in an approximate gradient (i.e., the echo gradient which may be deviated from the local stochastic gradient by a ratio $r$). Intuitively, in addition to the Byzantine tampering, we need to deal with non-determinism from stochastic gradients and noise from echo messages.

\subsection{Convergence Rate Analysis}
\label{s:rho}

In this part, we first analyze the \textit{convergence rate} $\rho$, which is a constant defined later in Equation (\ref{eq-rho}). Recall a few notations that $h=|\h|$ and $b=|\B|$, where given the execution, $\h$ is the set of fault-free workers and $\B$ is the set of Byzantine workers. Also recall that $L$ and $\mu$ are the constants defined in the Assumption \ref{as:lipschitz} and \ref{as:strong-convex}, respectively; $\sigma$ defined in Assumption \ref{as:bounded-variance}; and $r$ is the deviation ratio used in Echo-CGC. To derive $\rho$, we need to define series of constants based on the given parameters of $n, f, h, b, L, \mu, r$, and $\sigma$.

We first define a constant $\beta$ as
\begin{equation}
    \label{eq-beta}
    \beta = (n-2f)\frac{\mu-r(1+\sigma)L}{1+r} - b(1+k_h\sigma)L,
\end{equation}
where $k_x$ is defined as
\begin{equation}
    \label{eq-kn}
    k_x = 1+\frac{x-1}{\sqrt{2x-1}}, \enspace \forall x\geq 1.
\end{equation}
We then define a constant $\gamma$ as
\begin{equation}
    \label{eq-gamma}
    \gamma = nL^2\left(h(1+\sigma^2)+b\alpha_h\right),
\end{equation}
where
\begin{equation}
    \label{eq-alpha}
    \alpha_x = x\sigma^2 + (1+k_h\sigma)^2, \enspace \forall x\geq 1.
\end{equation}
Finally, we define the convergence rate $\rho$ using $\beta$ and $\gamma$ as follows:
\begin{equation}
    \label{eq-rho}
    \rho = 1-2\beta\eta+\gamma \eta^2.
\end{equation}

We will prove that under some standard assumptions,
the convergence rate $\rho$ is in the interval $[0,1)$. 
We first present several auxiliary lemmas. Due to page limit, most proofs are presented in Appendix \ref{app:rho}.

\begin{lemma}
    \label{lem:conv-1}
    Let $L,\mu>0$ be the Lipschitz constant and strong convexity constant defined in Assumption \ref{as:lipschitz} and \ref{as:strong-convex}, respectively. Then we have $\mu\leq L$.
\end{lemma}

\begin{lemma}
    \label{lem:kx}
    Denote $k^*=\sup_x\{k_x/\sqrt{x}:x\geq 1\}$. Then $k^*<\infty$, and numerically $k^*\approx 1.12$. Equivalently, $k_h\leq k^*\sqrt{h}$ for all $h\geq 1$.
\end{lemma}

\begin{lemma}
    \label{lem:conv-beta}
    Assume $n\mu-(3+k_n\sigma)fL>0$, then there exists $r>0$ that satisfies equation below.
    \begin{equation}
      \label{eq-beta-2}
       r < \frac{n\mu-(3+k_n\sigma)fL}{(n-2f)(1+\sigma)L+(1+k_n\sigma)fL}.
    \end{equation}
    Moreover, if $r>0$ satisfies Equation (\ref{eq-beta-2}), then $\beta>0$. 
\end{lemma}

Lemma \ref{lem:conv-beta} implies that we need to bound $\sigma$ for convergence. In general, Echo-CGC is correct if $\sigma = o(\log{n})$. For brevity, we make the following assumption to simplify the proof of convergence and the analysis of  communication complexity. We stress that this assumption can be relaxed using basically the same analysis with a denser mathematical manipulation.
\begin{assumption}
    \label{as:strict-sigma}
    Let $\sigma$ be the variance bound defined in Assumption \ref{as:bounded-variance}. We further assume that $\sigma < \frac{1}{\sqrt{n}}$.
\end{assumption}

Under Assumption \ref{as:strict-sigma}, we can narrow down the bound of $r$ in Lemma \ref{lem:conv-beta} to loosen our assumption on fault tolerance.

\begin{lemma}
    \label{lem:conv-strict-beta}
    Assume $n\mu-(3+k^*)fL>0$ ($k^*\approx1.12$), then there exists $r>0$ satisfying Equation (\ref{eq-deviation-ratio}) such that $\beta >0$.
    \begin{equation}
        \label{eq-deviation-ratio}
        r < \frac{n\mu-(3+k^*)fL}{(n-2f)(1+\sigma)L+(1+k^*)fL}.
    \end{equation}
\end{lemma}

\begin{theorem}
    \label{thm:conv-rate}
    Assume $n\mu-(3+k^*)fL>0$ and $r$ is a value that satisfies Inequality (\ref{eq-deviation-ratio}). Then we can find an $\eta > 0$ such that $\eta < 2\beta/\gamma$,
    which in turn makes $\rho\in[0,1)$.
\end{theorem}

\subsection{Proof of Convergence}
\label{s:converge}

Next, we prove the convergence of our algorithm. That is, Echo-CGC converges to the optimal point $w^*$ of the cost function $Q$.
We prove the convergence under the assumption that $n\mu-(3+k^*)fL>0$.
Due to page limit, we present key proofs here, and the rest can be found in Appendix \ref{app:converge}.

Recall our definition of the conditional expectation $\EE=\EE_{\Xi^t}(\cdot\mid w^t, G_\B^t)$ introduced in Section \ref{sec:assumptions}. Before proving the main theorem, we introduce some preliminary lemmas.

\begin{lemma}
\label{lem:eq-os-4}
For all $t$ and for all $j\in \h$,
\begin{equation}
\label{eq-os-4}
\EE\norm{g_j^t} \leq (1+\sigma)\norm{\nabla Q(w^t)}.
\end{equation}
\end{lemma}

\begin{lemma}
    \label{lem:order-stats}
    Recall that $\hat{g}_j^t$ is the gradient after applying the CGC filter. For all $t$ and for all $j\in\{1,2,\ldots,n\}$,
    \begin{equation}
        \label{eq-order-stat}
        \EE\norm{\hat{g}_j^t} \leq (1+k_h\sigma)\norm{\nabla Q(w^t)}.
    \end{equation}
\end{lemma}

The proof of Lemma \ref{lem:order-stats} is based on Lemma \ref{lem:eq-os-4} and the following prior results: Gumbel \cite{Gumbel-order-stats1} and Hartley and David \cite{Hartley-David-OrderStats2} proved that given identical means and variances $(\mu,\sigma^2)$, the upper bound of the expectation of the largest random variable among $n$ independent random variables is $\mu+\frac{\sigma(n-1)}{\sqrt{2n-1}}$.

\begin{lemma}
    \label{lem:order-stat2}
    Following the same setup, for all $t$ and for all $j\in\{1,2,\ldots,n\}$,
    \begin{equation}
        \EE\norm{\hat{g}_j^t}^2 \leq \alpha_h\norm{\nabla Q(w^t)}^2.
    \end{equation}
\end{lemma}

The proof of Lemma \ref{lem:order-stat2} is based on Lemma \ref{lem:eq-os-4} and the following result: Papadatos \cite{Papdatos-max-variance-OS} proved that for $n$ i.i.d. random variables $X_1\leq X_2\leq \dotsm\leq X_n$ with finite variance $\sigma^2$, the maximum variance of $X_n$ is bounded above by $n\sigma^2$.

Lemma \ref{lem:order-stats} and Lemma \ref{lem:order-stat2} provide upper bounds on $\EE\norm{\hat{g}_j^t}$ and $\EE\norm{\hat{g}_j^t}^2$. 
These two bounds allow us to bound the impact of bogus gradients transmitted by a faulty node $j$. If $j$ transmitted an extreme gradient, it would be dropped by the CGC filter; otherwise, these two bounds essentially imply that the filtered gradient $\hat{g}_j^t$ has some nice property even if $j$ is faulty. 
For fault-free gradients, Lemma \ref{lem:eq-os-4} provides a better bound.

\begin{theorem}
    \label{thm:convergent}
    Assume that $n\mu-(3+k^*)fL>0$. We can find $r>0$ that satisfies Inequality (\ref{eq-deviation-ratio}) and $\eta>0$ such that $\eta<2\beta/\gamma$. Echo-CGC with the chosen $r$ and $\eta$ will converge to the optimal parameter $w^*$ as $t\to\infty$.
\end{theorem}

\begin{proof}
    Our ultimate goal is to show that the sequence $\{\EE\norm{w^t-w^*}^2\}_{t=0}^\infty$ converges to $0$. Recall that the aggregation rule of the algorithm is $w^{t+1}=w^t-\eta g^t$. Thus, we obtain that
    \begin{align}
        \EE\norm{w^{t+1}-w^*}^2 
        &\leq \EE\norm{w^t-w^*-\eta g^t}^2 \notag \\
        &= \underbrace{\EE\norm{w^t-w^*}^2}_A
        - \underbrace{2\eta\EE\dotprod{w^t-w^*}{g^t}}_B 
        + \underbrace{\eta^2\EE\norm{g^t}^2}_C.
    \end{align}
    Since $w^t$ is known, $w^t$ can be treated as a constant, and $\EE\norm{w^t-w^*}^2=\norm{w^t-w^*}^2$.
    
    \textbf{Part $C$}: In Appendix \ref{app:part-C}, we show that the following inequality holds.
    \begin{align}
        \label{eq-C3}
        \EE\norm{g^t}^2 \leq \gamma \norm{w^t-w^*}^2.
    \end{align}
    
    \textbf{Part $B$}: By linearity of inner product,
    \begin{align}
        \label{eq-B0}
        \dotprod{w^t-w^*}{g^t} 
        &= \sum_{j=1}^n\dotprod{w^t-w^*}{\hat{g}_j^t} \notag\\
        &= \sum_{j\in\h}\dotprod{w^t-w^*}{\hat{g}_j^t} + \sum_{j\in\B}\dotprod{w^t-w^*}{\hat{g}_j^t}.
    \end{align}
    First, by Schwarz Inequality, $\dotprod{w^t-w^*}{\hat{g}_j^t}\geq -\norm{w^t-w^*}\norm{\hat{g}_j^t}$; by Lemma \ref{lem:order-stats} and $L$-Lipschitz assumption, $\EE\norm{\hat{g}_j^t}\leq (1+k_h\sigma)L\norm{w^t-w^*}$. Thus,
    \begin{equation}
        \label{eq-Bn0}
        \EE\dotprod{w^t-w^*}{\hat{g}_j^t} 
        \geq -(1+k_h\sigma) L\norm{w^t-w^*}^2, \enspace \forall j\in\B.
    \end{equation}
    Next, observe that by our algorithm, for each $j\in\h$, the received gradient before CGC filter $\tilde{g}_j^t$ satisfies (i) $\norm{\tilde{g}_j^t}=\norm{g_j^t}$ and (ii) $\tilde{g}_j^t=a_j(g_j^t+\Delta g_j^t)$, for some constant $a_j=\norm{g_j^t}/\norm{g_j^t+\Delta g_j^t}$ and a vector $\Delta g_j^t$ 
    such that $\norm{\Delta g_j^t}\leq r\norm{g_j^t}$. This implies $a_j\geq 1/(1+r)$. Therefore,
    \begin{align}
        \label{eq-Bn1}
        \EE\dotprod{w^t-w^*}{\tilde{g}_j^t} 
        &= \EE\dotprod{w^t-w^*}{a_j(g_j^t+\Delta g_j^t)} \notag\\
        &\geq \frac{1}{1+r}\left(\dotprod{w^t-w^*}{\EE g_j^t} + \EE\dotprod{w^t-w^*}{\Delta g_j^t}\right), \enspace \forall j\in\h.
    \end{align}
    By Assumption \ref{as:iid}, $\EE g_j^t=\nabla Q(w^t)$; by strong convexity,
    \[
    \dotprod{w^t-w^*}{\nabla Q(w^t)}\geq \mu\norm{w^t-w^*}^2.
    \]
    By Schwarz inequality, $\EE\dotprod{w^t-w^*}{\Delta g_j^t} \geq -\norm{w^t-w^*}\EE\norm{\Delta g_j^t}$; and $\EE\norm{\Delta g_j^t}\leq r\EE\norm{g_j^t}$.
    By Lemma \ref{lem:eq-os-4} and $L$-Lipschitz assumption, $\EE\norm{g_j^t}\leq (1+\sigma)L\norm{w^t-w^*}$. Thus,
    \begin{equation}
        \EE\dotprod{w^t-w^*}{\Delta g_j^t} \geq -r(1+\sigma)L\norm{w^t-w^*}^2. \notag
    \end{equation}
    Upon substituting these results into Equation (\ref{eq-Bn1}), we obtain that
    \begin{equation}
        \label{eq-Bn2}
        \EE\dotprod{w^t-w^*}{\tilde{g}_j^t} \geq \frac{\mu-r(1+\sigma)L}{1+r}\norm{w^t-w^*}^2, \enspace \forall j\in\h.
    \end{equation}
    
    We partition $\h$ into two parts: $\h_1=\h\cap\{i_1,\ldots,i_{n-f}\}$ and $\h_2=\h\setminus\h_1$. For each $j\in\h_1$, the received gradient is unchanged by CGC filter, i.e., $\hat{g}_j^t=\tilde{g}_j^t$. Therefore, Equation (\ref{eq-Bn2}) also holds for $\hat{g}_j^t$, for all $j\in\h_1$.
    
    The case of $\h_2$ is similar. Note that for each $j\in\h_2$, the gradient $\tilde{g}_j^t$ is scaled down to $\hat{g}_j^t$ by CGC filter. In other words, there exists some constant $a_j'\geq 0$ such that $\hat{g}_j^t=a_j'\tilde{g}_j^t$. Therefore, by Equation (\ref{eq-Bn1}),
    \begin{equation}
        \EE\dotprod{w^t-w^*}{\hat{g}_j^t} 
        = \EE\dotprod{w^t-w^*}{a_j'\tilde{g_j^t}}
        = a_j'\EE\dotprod{w^t-w^*}{\tilde{g}_j^t}, \enspace \forall j\in\h_2. \notag
    \end{equation}
    We can verify that if by assumption that $r>0$ satisfies Equation (\ref{eq-deviation-ratio}), then $\mu-r(1+\sigma)L>0$; and Equation (\ref{eq-Bn1}) implies that $\EE\dotprod{w^t-w^*}{\tilde{g}_j^t}\geq 0$. Therefore,
    \begin{equation}
        \label{eq-Bn3}
        \EE\dotprod{w^t-w^*}{\hat{g}_j^t} \geq 0, \enspace \forall j\in\h_2.
    \end{equation}
    Note that $|\h_1|\geq h-2f$. Upon substituting Equation (\ref{eq-Bn0}), (\ref{eq-Bn2}), (\ref{eq-Bn3}) into Equation (\ref{eq-B0}), we obtain that
    \begin{align}
        \label{eq-B5}
        \EE\dotprod{w^t-w^*}{g^t}
        &\geq \left( (n-2f)\frac{\mu-r(1+\sigma)L}{1+r} - b(1+k_h\sigma)L \right)\norm{w^t-w^*}^2.
    \end{align}
    By definition of $\beta$ in Equation (\ref{eq-beta}), this implies $\EE\dotprod{w^t-w^*}{g^t}\geq \beta\norm{w^t-w^*}^2$.
    
    \textbf{Conclusion:}
    Upon combining part $A,B$ and $C$, by definition of $\rho$ in Equation (\ref{eq-rho}),
    \begin{equation}
        \EE\norm{w^{t+1}-w^*}^2 \leq \rho\norm{w^t-w^*}^2, \enspace \forall t=0,1,2,\ldots \notag
    \end{equation}
    Recall the definition of the conditional expectation operator $\EE$. This implies that
    \begin{equation}
        \EE(\norm{w^t-w^*}^2\mid w^0,\mathcal{G}_\B^0,\ldots,\mathcal{G}_\B^t) \leq \rho^t\norm{w^0-w^*}^2 \notag
    \end{equation}
    By Theorem \ref{thm:conv-rate}, $\rho\in(0,1)$. Therefore, as $t\to\infty$, $\norm{w^t-w^*}^2$ converges to $0$. In other words, $w^t$ converges to the optimal parameter $w^*$. This proves the theorem.
\end{proof}

\subsection{Communication Complexity}
\label{s:communication}
We analyze the communication complexity of the Echo-CGC algorithm, and show that under suitable conditions, it effectively reduces communication complexity compared to prior algorithms \cite{Rachid_Krum, vaidya-2f-redundancy}. 
First consider a ball in $\RR^d$ whose center is the true gradient $\nabla Q(w^t)$:
\begin{equation}
    B(\nabla Q(w^t), \frac{r}{2+r}\norm{\nabla Q(w^t)}) = \{u\in\RR^d:\norm{u-\nabla Q(w^t)}\leq \frac{r}{2+r}\norm{\nabla Q(w^t)}\},
\end{equation}
where $r>0$ is the deviation ratio. For a slight abuse of notations, we abbreviate the ball as $B$. This should not be confused with $\B$, the set of Byzantine workers. We present only the main results, and the proofs can be found in Appendix \ref{app:communication}.

\begin{lemma}
    \label{lem:ball}
    For all $u,v\in B$, $\norm{u-v}\leq r\norm{u}$ (and $\norm{u-v}\leq r\norm{v}$).
\end{lemma}

Given Lemma \ref{lem:ball}, we compute the probability that an arbitrary gradient $g_j^t$ is in the ball $B$.
By Markov's Inequality,
\begin{align}
    \prob(g_j^t\in B)
    &= \prob\left(\norm{g_j^t-\nabla Q(w^t)}^2\leq \frac{r^2}{(2+r)^2}\norm{\nabla Q(w^t)}^2\right) \notag\\
    &\geq 1 - \frac{\EE\norm{g_j^t-\nabla Q(w^t)}^2}{\frac{r^2}{(2+r)^2}\norm{\nabla Q(w^t)}^2}.
\end{align}
By Assumption \ref{as:bounded-variance}, $\EE\norm{g_j^t-\nabla Q(w^t)}^2\leq \sigma^2\norm{\nabla Q(w^t)}^2$, so we conclude that $\prob(g_j^t\in B) \geq p$, where $p$ is the lower bound defined as $p=1-(1+2/r)^2\sigma^2$.

Denote $n_B=|\{j:g_j^t\in B\}|$ and $n^*$ as the number of workers that send the ``echo message'' in a round. By Lemma \ref{lem:ball}, $n^*\geq n_B-1$.  Since each event $\{g_j^t\in B\}$ is independent and has a fixed probability, $n^*$ follows a Binomial distribution with success probability $\prob(g_j^t\in B)$ which is bounded below by $p$. Therefore, 
\begin{equation}
    \EE n^*\geq \EE n_B - 1 \geq np-1. \notag
\end{equation}
For $n\gg 1$, we assume that $1/n\approx 0$. Also in practice, $d\gg n$, so the message complexity of each echo message (in $O(n)$ bits) is negligible compared to raw gradients (in $O(d)$ bits). Hence, the ratio of bit complexity of our algorithm and prior algorithms (e.g., \cite{Rachid_Krum,vaidya-2f-redundancy}) can be approximately bounded above as follows:
\begin{align}
    \frac{\text{bit complexity of Echo-CGC}}{\text{bit complexity of prior algorithms}} 
    &= \frac{n^*O(n)+(n-n^*)O(d)}{nO(d)} \notag\\
    &\leq \frac{(np-1)O(n)+[n-(np-1)]O(d)}{nO(d)} \notag\\
    &\approx 1-p. \notag
\end{align}
We denote the upper bound of ratio of reduced complexity to complexity of prior algorithms as $C=1-p=(1+2/r)^2\sigma^2$.

\textbf{Analysis}~~By Equation (\ref{lem:conv-beta}) and Lemma \ref{lem:kx}, $C$ can be expressed as
\begin{equation}
    \label{eq:complexity-C}
    C \leq \sigma^2\left(1+2\cdot\frac{(1-2x)(1+\sigma)+(1+\sigma k^*\sqrt{n})x}{\mu/L-(3+\sigma k^*\sqrt{n})x}\right)^2,
\end{equation}
where $x=f/n$ is the fault-tolerance factor.

As Equation (\ref{eq:complexity-C}) shows, the ratio $\C$ is related to four non-trivial variables: (i) bound of variance $\sigma\geq 0$; (ii) resilience $x=f/n$ satisfying the assumption in Lemma \ref{lem:conv-beta}, i.e.,
\[
\mu/L - (3+\sigma k^*\sqrt{n})x > 0;
\]
(iii) constant $L/\mu$, which is determined by the cost function $Q$ and satisfies $0<L/\mu<1$ by Lemma \ref{lem:conv-1}; and (iv) number of workers $n>0$. 

\begin{figure}[]
    \centering
    \begin{subfigure}[b]{0.45\textwidth}
        \centering
        \begin{tikzpicture}[scale=0.75, label distance=2em]
            \begin{axis}[
                axis lines = left, xlabel = $\sigma$, ylabel = $\C$, ymin=0, ymax=1, xtick={0,0.04,0.08,0.12,0.16},
            ]
                \addplot[domain=0:0.2,color=blue] {x^2 * ( 1 + 2 * ((1-2*0.1)*(1+x)+(1+1.12*x*sqrt(100))*0.1) / (1-(3+1.12*x*sqrt(100))*0.1) )^2}; 
            \end{axis}
        \end{tikzpicture}
        \caption{$C$ as a function of $\sigma$, for fixed $\mu/L=1$, $x=0.1$, and $n=100$.}
        \label{fig:C-vs-sigma}
    \end{subfigure}
    \hfill
    \begin{subfigure}[b]{0.45\textwidth}
        \centering
        \begin{tikzpicture}[scale=0.75]
            \begin{axis}[
                axis lines = left, xlabel = $\mu/L$, ylabel = $\C$, ymin=0, ymax=1, xmax=1,
            ]
                \addplot[domain=0.5:1,color=blue] {0.1^2 * ( 1 + 2 * ((1-2*0.1)*(1+0.1)+(1+1.12*0.1*sqrt(100))*0.1) / (x-(3+1.12*0.1*sqrt(100))*0.1) )^2}; 
            \end{axis}
        \end{tikzpicture}
        \caption{$C$ as a function of $\mu/L$, for fixed $\sigma=0.1$, $x=0.1$, and $n=100$.}
        \label{fig:C-vs-mu/L}
    \end{subfigure}
    \begin{subfigure}[b]{0.45\textwidth}
        \centering
        \begin{tikzpicture}[scale=0.75]
            \begin{axis}[
                axis lines = left, xlabel = $x$, ylabel = $\C$, ymin=0, ymax=1, xmin=0,
            ]
                \addplot[domain=0:0.2,color=blue] {0.1^2 * ( 1 + 2 * ((1-2*x)*(1+0.1)+(1+1.12*0.1*sqrt(100))*x) / (1-(3+1.12*0.1*sqrt(100))*x) )^2}; 
            \end{axis}
        \end{tikzpicture}
        \caption{$C$ as a function of $x$, for fixed $\sigma=0.1$, $\mu/L=1$, and $n=100$.}
        \label{fig:C-vs-x}
    \end{subfigure}
    \hfill
    \begin{subfigure}[b]{0.45\textwidth}
        \centering
        \begin{tikzpicture}[scale=0.75]
            \begin{axis}[
                axis lines = left, xlabel = $n$, ylabel = $\C$, ymin=0, ymax=1, xtick={0,100,200,300,400,500}, xmin=0, xmax=500,
            ]
                \addplot[domain=0:500,color=blue] {0.1^2 * ( 1 + 2 * ((1-2*0.1)*(1+0.1)+(1+1.12*0.1*sqrt(x))*0.1) / (1-(3+1.12*0.1*sqrt(x))*0.1) )^2}; 
            \end{axis}
        \end{tikzpicture}
        \caption{$C$ as a function of $n$, for fixed $\sigma=0.1$, $\mu/L=1$, and $x=0.1$.}
        \label{fig:C-vs-n}
    \end{subfigure}
\end{figure}

We first plot the relation between one factor and $C$ while fixing the other three factors. 
First, we present the most significant fact, $\sigma$. 
We fix $\mu/L=1$, $x=0.1$, and $n=100$. As Figure \ref{fig:C-vs-sigma} shows, $C$ increases in an almost quadratic speed with $\sigma$ because of the $\sigma^2$ term in Equation (\ref{eq:complexity-C}). Therefore, our algorithm is guaranteed to have lower communication complexity when the variance of gradients is relatively low, especially when $\sigma\leq 0.1$. In practice, this is the scenario when the data set consists mainly of similar data instances. 

Then, we plot $C$ against $\mu/L$ with fixed $\sigma=0.1$, $x=0.1$, and $n=100$. As Figure \ref{fig:C-vs-mu/L} shows, $C$ decreases as $\mu/L$ becomes closer to $1$. As $\mu/L>0.75$, $C<0.5$, meaning that $[0.75,1]$ is the range of $\mu/L$ where our algorithm is guaranteed to perform significantly better.

Next, we plot $C$ against $x$ with fixed $\sigma=0.1$, $\mu/L$, and $n=100$. As Figure \ref{fig:C-vs-x} shows, there is a trade-off between $C$ and fault resilience $x$. As $x$ approaches the max resilience defined in Lemma \ref{lem:conv-beta}, i.e., $x_{\max}=\frac{\mu/L}{(3+\sigma k^*\sqrt{n}}$, the theoretical upper bound $C$ blows up. Moreover, as $x<0.15$, $C<0.4$; and thus $[0,0.15]$ is a proper range of $x$.

Finally, we plot $C$ against $n$ with fixed $\sigma=0.1$, $\mu/L=1$, and $x=0.1$. As Figure \ref{fig:C-vs-n} shows, $C$ increases almost linearly with respect to $n$ with a relatively flat slope. In other words, $n$ is \textit{not} a significant factor of $C$; and the performance of our algorithm is stable in a wide range of $n$.

In conclusion, our algorithm is guaranteed to require lower communication complexity when: (i) $\sigma$ is low, i.e., data instances are similar and (ii) $\mu/L$ is close to $1$. Also, there is a trade-off between resilience and efficiency. As a concrete example, when $\sigma=0.1$, $x=0.2$, $\mu/L=1$, and $n=100$, $C\approx0.25$, meaning that our  algorithm is guaranteed to save at least $75\%$ of communication cost.

\section{Summary}

In this paper, we present our Byzantine-tolerant DML algorithm that incurs lower communication complexity in a single-hop radio netowrk (under suitable conditions). Our algorithm is inspired by the CGC filter \cite{vaidya-2f-redundancy}, but we need to devise new proofs to handle the randomness and noise introduced in our mechanism.

There are two interesting open problems: (i) multi-hop radio network; and (ii) different mechanism for constructing echo messages, e.g., usage of angles rather than distance ratio.

\bibliography{paperlist}

\begin{thebibliography}{10}

\bibitem{Broadcast_SPAA10}
Dan Alistarh, Seth Gilbert, Rachid Guerraoui, Zarko Milosevic, and Calvin
  Newport.
\newblock Securing every bit: Authenticated broadcast in radio networks.
\newblock In {\em Proceedings of the Twenty-Second Annual ACM Symposium on
  Parallelism in Algorithms and Architectures}, SPAA '10, page 50–59, New
  York, NY, USA, 2010. Association for Computing Machinery.
\newblock URL: \url{https://doi.org/10.1145/1810479.1810489}, \href
  {http://dx.doi.org/10.1145/1810479.1810489}
  {\path{doi:10.1145/1810479.1810489}}.

\bibitem{JAM_PODC08}
Baruch Awerbuch, Andrea Richa, and Christian Scheideler.
\newblock A jamming-resistant mac protocol for single-hop wireless networks.
\newblock In {\em Proceedings of the Twenty-Seventh ACM Symposium on Principles
  of Distributed Computing}, PODC '08, page 45–54, New York, NY, USA, 2008.
  Association for Computing Machinery.
\newblock URL: \url{https://doi.org/10.1145/1400751.1400759}, \href
  {http://dx.doi.org/10.1145/1400751.1400759}
  {\path{doi:10.1145/1400751.1400759}}.

\bibitem{Broadcast_Nitin_PODC05}
Vartika Bhandari and Nitin~H. Vaidya.
\newblock On reliable broadcast in a radio network.
\newblock In {\em Proceedings of the Twenty-Fourth Annual ACM Symposium on
  Principles of Distributed Computing}, PODC '05, page 138–147, New York, NY,
  USA, 2005. Association for Computing Machinery.
\newblock URL: \url{https://doi.org/10.1145/1073814.1073841}, \href
  {http://dx.doi.org/10.1145/1073814.1073841}
  {\path{doi:10.1145/1073814.1073841}}.

\bibitem{Rachid_Krum}
Peva Blanchard, El~Mahdi El~Mhamdi, Rachid Guerraoui, and Julien Stainer.
\newblock Machine learning with adversaries: Byzantine tolerant gradient
  descent.
\newblock NIPS'17, page 118–128, Red Hook, NY, USA, 2017. Curran Associates
  Inc.

\bibitem{convex-optimization}
Stephen Boyd and Lieven Vandenberghe.
\newblock {\em Convex Optimization}.
\newblock Cambridge University Press, USA, 2004.

\bibitem{Su_BGD}
Yudong Chen, Lili Su, and Jiaming Xu.
\newblock Distributed statistical machine learning in adversarial settings:
  Byzantine gradient descent.
\newblock {\em Proc. ACM Meas. Anal. Comput. Syst.}, 1(2), December 2017.
\newblock URL: \url{https://doi.org/10.1145/3154503}, \href
  {http://dx.doi.org/10.1145/3154503} {\path{doi:10.1145/3154503}}.

\bibitem{Rachid_Kardam}
Georgios Damaskinos, El~Mahdi El~Mhamdi, Rachid Guerraoui, Rhicheek Patra, and
  Mahsa Taziki.
\newblock Asynchronous {B}yzantine machine learning (the case of {SGD}).
\newblock volume~80 of {\em Proceedings of Machine Learning Research}, pages
  1145--1154, Stockholmsmässan, Stockholm Sweden, 10--15 Jul 2018. PMLR.
\newblock URL: \url{http://proceedings.mlr.press/v80/damaskinos18a.html}.

\bibitem{Rachid_genuine_BML}
El-Mahdi El-Mhamdi, Rachid Guerraoui, Arsany Guirguis, L\^{e}~Nguy\^{e}n Hoang,
  and S\'{e}bastien Rouault.
\newblock Genuinely distributed byzantine machine learning.
\newblock In {\em Proceedings of the 39th Symposium on Principles of
  Distributed Computing}, PODC '20, page 355–364, New York, NY, USA, 2020.
  Association for Computing Machinery.
\newblock URL: \url{https://doi.org/10.1145/3382734.3405695}, \href
  {http://dx.doi.org/10.1145/3382734.3405695}
  {\path{doi:10.1145/3382734.3405695}}.

\bibitem{Medical_ML}
J.~{Fan} and J.~{Lv}.
\newblock A selective overview of variable selection in high dimensional
  feature space.
\newblock {\em Statistica Sinica}, pages 101 -- 148, 01 2010.

\bibitem{Gumbel-order-stats1}
E.~J. Gumbel.
\newblock The maxima of the mean largest value and of the range.
\newblock {\em The Annals of Mathematical Statistics}, 25(1):76--84, 1954.
\newblock URL: \url{http://www.jstor.org/stable/2236513}.

\bibitem{vaidya-2f-redundancy}
Nirupam Gupta and Nitin~H. Vaidya.
\newblock Fault-tolerance in distributed optimization: The case of redundancy.
\newblock In {\em Proceedings of the 39th Symposium on Principles of
  Distributed Computing}, PODC '20, page 365–374, New York, NY, USA, 2020.
  Association for Computing Machinery.
\newblock URL: \url{https://doi.org/10.1145/3382734.3405748}, \href
  {http://dx.doi.org/10.1145/3382734.3405748}
  {\path{doi:10.1145/3382734.3405748}}.

\bibitem{Hartley-David-OrderStats2}
H.~O. Hartley and H.~A. David.
\newblock Universal bounds for mean range and extreme observation.
\newblock {\em The Annals of Mathematical Statistics}, 25(1):85--99, 1954.
\newblock URL: \url{http://www.jstor.org/stable/2236514}.

\bibitem{Wireless_FL_Chiang_2020}
Seyyedali Hosseinalipour, Christopher~G. Brinton, Vaneet Aggarwal, Huaiyu Dai,
  and Mung Chiang.
\newblock From federated learning to fog learning: Towards large-scale
  distributed machine learning in heterogeneous wireless networks, 2020.
\newblock \href {http://arxiv.org/abs/2006.03594} {\path{arXiv:2006.03594}}.

\bibitem{Broadcast_Koo_PODC04}
Chiu-Yuen Koo.
\newblock Broadcast in radio networks tolerating byzantine adversarial
  behavior.
\newblock In {\em Proceedings of the Twenty-Third Annual ACM Symposium on
  Principles of Distributed Computing}, PODC '04, page 275–282, New York, NY,
  USA, 2004. Association for Computing Machinery.
\newblock URL: \url{https://doi.org/10.1145/1011767.1011807}, \href
  {http://dx.doi.org/10.1145/1011767.1011807}
  {\path{doi:10.1145/1011767.1011807}}.

\bibitem{NIPS_Li_2014}
Mu~Li, David~G. Andersen, Alexander Smola, and Kai Yu.
\newblock Communication efficient distributed machine learning with the
  parameter server.
\newblock In {\em Proceedings of the 27th International Conference on Neural
  Information Processing Systems - Volume 1}, NIPS'14, page 19–27, Cambridge,
  MA, USA, 2014. MIT Press.

\bibitem{ACM_Survey_DL_2020}
Ruben Mayer and Hans-Arno Jacobsen.
\newblock Scalable deep learning on distributed infrastructures: Challenges,
  techniques, and tools.
\newblock {\em ACM Comput. Surv.}, 53(1), February 2020.
\newblock URL: \url{https://doi.org/10.1145/3363554}, \href
  {http://dx.doi.org/10.1145/3363554} {\path{doi:10.1145/3363554}}.

\bibitem{JAM_Infocom07}
V.~{Navda}, A.~{Bohra}, S.~{Ganguly}, and D.~{Rubenstein}.
\newblock Using channel hopping to increase 802.11 resilience to jamming
  attacks.
\newblock In {\em IEEE INFOCOM 2007 - 26th IEEE International Conference on
  Computer Communications}, pages 2526--2530, 2007.

\bibitem{Papdatos-max-variance-OS}
Nickos Papadatos.
\newblock Maximum variance of order statistics.
\newblock {\em Annals of the Institute of Statistical Mathematics},
  47:185--193, 02 1995.
\newblock \href {http://dx.doi.org/10.1007/BF00773423}
  {\path{doi:10.1007/BF00773423}}.

\bibitem{PERLMAN197452}
Michael~D. Perlman.
\newblock Jensen's inequality for a convex vector-valued function on an
  infinite-dimensional space.
\newblock {\em Journal of Multivariate Analysis}, 4(1):52 -- 65, 1974.
\newblock URL:
  \url{http://www.sciencedirect.com/science/article/pii/0047259X74900050},
  \href {http://dx.doi.org/https://doi.org/10.1016/0047-259X(74)90005-0}
  {\path{doi:https://doi.org/10.1016/0047-259X(74)90005-0}}.

\bibitem{vaidya_optimization_PODC16}
Lili Su and Nitin~H. Vaidya.
\newblock Fault-tolerant multi-agent optimization: Optimal iterative
  distributed algorithms.
\newblock In George Giakkoupis, editor, {\em Proceedings of the 2016 {ACM}
  Symposium on Principles of Distributed Computing, {PODC} 2016, Chicago, IL,
  USA, July 25-28, 2016}, pages 425--434. {ACM}, 2016.
\newblock URL: \url{https://doi.org/10.1145/2933057.2933105}, \href
  {http://dx.doi.org/10.1145/2933057.2933105}
  {\path{doi:10.1145/2933057.2933105}}.

\bibitem{vaidya_nonBayesian_learning}
Lili Su and Nitin~H. Vaidya.
\newblock Non-bayesian learning in the presence of byzantine agents.
\newblock In {\em Distributed Computing - 30th International Symposium, {DISC}
  2016, Paris, France, September 27-29, 2016. Proceedings}, pages 414--427,
  2016.
\newblock URL: \url{https://doi.org/10.1007/978-3-662-53426-7\_30}, \href
  {http://dx.doi.org/10.1007/978-3-662-53426-7\_30}
  {\path{doi:10.1007/978-3-662-53426-7\_30}}.

\bibitem{IEEE_Survey_ML_Wireless}
Y.~{Sun}, M.~{Peng}, Y.~{Zhou}, Y.~{Huang}, and S.~{Mao}.
\newblock Application of machine learning in wireless networks: Key techniques
  and open issues.
\newblock {\em IEEE Communications Surveys Tutorials}, 21(4):3072--3108, 2019.

\bibitem{HotEdge_Tao_2018}
Zeyi Tao and Qun Li.
\newblock esgd: Communication efficient distributed deep learning on the edge.
\newblock In {\em {USENIX} Workshop on Hot Topics in Edge Computing (HotEdge
  18)}, Boston, MA, July 2018. {USENIX} Association.
\newblock URL:
  \url{https://www.usenix.org/conference/hotedge18/presentation/tao}.

\bibitem{ACM_Survey_DML_2020}
Joost Verbraeken, Matthijs Wolting, Jonathan Katzy, Jeroen Kloppenburg, Tim
  Verbelen, and Jan~S. Rellermeyer.
\newblock A survey on distributed machine learning.
\newblock {\em ACM Comput. Surv.}, 53(2), March 2020.
\newblock URL: \url{https://doi.org/10.1145/3377454}, \href
  {http://dx.doi.org/10.1145/3377454} {\path{doi:10.1145/3377454}}.

\bibitem{Xie_Zeno}
Cong Xie, Sanmi Koyejo, and Indranil Gupta.
\newblock Zeno: Distributed stochastic gradient descent with suspicion-based
  fault-tolerance.
\newblock volume~97 of {\em Proceedings of Machine Learning Research}, pages
  6893--6901, Long Beach, California, USA, 09--15 Jun 2019. PMLR.
\newblock URL: \url{http://proceedings.mlr.press/v97/xie19b.html}.

\end{thebibliography}

\newpage
\appendix

\section{Proof of Lemmas in Section \ref{s:rho}}
\label{app:rho}
\subsection{Proof of Lemma \ref{lem:conv-1}}

\begin{proof}
    Cauchy inequality and $L$-Lipschitz smoothness of cost function imply that
    \begin{equation}
        \label{eq-rho1}
        \dotprod{w-w^*}{\nabla Q(w)} \leq \norm{w-w^*}\norm{\nabla Q(w)} \leq L\norm{w-w^*}^2, \enspace \forall w\in\RR^d.
    \end{equation}
    Also by strong convexity,
    \begin{equation}
        \label{eq-rho2}
        \dotprod{w-w^*}{\nabla Q(w)} \geq 
        \mu\norm{w-w^*}^2, \enspace \forall w\in\RR^d.
    \end{equation}
    Equation (\ref{eq-rho1}) and (\ref{eq-rho2}) together imply that $\mu\leq L$.
\end{proof}

\subsection{Proof of Lemma \ref{lem:kx}}

\begin{proof}
    First, we check the derivative of $k_x/\sqrt{x}$:
    \begin{equation}
        \frac{d}{dx}\frac{k_x}{\sqrt{x}} = -\frac{(2x-1)^{3/2}-3x+1}{2(2x^2-x)^{3/2}}. \notag
    \end{equation}
    We can then verify numerically that $k_x/\sqrt{x}$ reaches maximum at $\approx1.91$ and that $k^*=\sup_{x\geq 1}(k_x/\sqrt{x})\approx 1.12$.
\end{proof}

\subsection{Proof of Lemma \ref{lem:conv-beta}}

\begin{proof}
    By definition of $k_x$, $k_x=1+\frac{x-1}{\sqrt{2x-1}}$ is increasing for $x\geq 1$. Recall that $h\leq n$ and $b\leq f$, so $k_h\leq k_n$. By definition of $\beta$, this further implies that
    \begin{equation}
        \beta \geq (n-2f)\frac{\mu-r(1+\sigma)L}{1+r}-fL(1+k_n\sigma). \notag
    \end{equation}
    Therefore, $\beta>0$ if
    \begin{align}
        &(n-2f)\left(\mu-r(1+\sigma)L\right) > (1+r)(1+k_n\sigma)fL \notag\\
        \iff & (n-2f)\mu-(1+k_n\sigma)fL > \left( (n-2f)(1+\sigma)L+(1+k_n\sigma)fL \right) r. \notag
    \end{align}
    Recall that by Lemma \ref{lem:conv-1}, $\mu\leq L$, so the above inequality is true if
    \begin{equation}
        \label{eq:app-r-sufficient}
        n\mu - (3+k_n\sigma)fL > \left( (n-2f)(1+\sigma)L+(1+k_n\sigma)fL \right) r.
    \end{equation}
    Since we assume $n-2f>0$, the right side is always positive. Therefore, if $n\mu-(3+k_n\sigma)fL>0$, then there exists $r>0$ that satisfies Equation (\ref{eq:app-r-sufficient}); also such $r$ satisfies $\beta>0$. This proves the lemma.
\end{proof}

\subsection{Proof of Lemma \ref{lem:conv-strict-beta}}

\begin{proof}
    By Lemma \ref{lem:kx}, $k_n\leq k^*\sqrt{n}$. Given the assumption that $\sigma<1/\sqrt{n}$, 
    \[
    \frac{n\mu - (3+k_n\sigma)fL}{(n-2f)(1+\sigma)L+(1+k_n\sigma)fL} > \frac{n\mu - (3+k^*)fL}{(n-2f)(1+\sigma)L+(1+k^*)fL}.
    \]
    Therefore, if $r>0$ satisfies Equation (\ref{eq-deviation-ratio}), then it also satisfies Equation (\ref{eq-beta-2}); and by Lemma \ref{lem:conv-strict-beta}, such $r$ satisfies $\beta>0$.
\end{proof}

\subsection{Proof of Theorem \ref{thm:conv-rate}}

\begin{proof}
    First observe that $\rho$ can be represented as a quadratic function of $\eta$, i.e., $\rho(\eta)=\gamma\eta^2-2\beta\eta+1$. We prove the theorem by providing the bounds on this function.
    
    By definition of $\alpha_x$ in Equation (\ref{eq-alpha}), $\alpha_h>0$ for all $h\geq 1$. Thus, by definition of $\gamma$ in Equation (\ref{eq-gamma}), $\gamma>0$. Also, by Lemma \ref{lem:conv-strict-beta}, for $r>0$ that satisfies Equation (\ref{eq-deviation-ratio}), $\beta>0$.
    
    Recall that for a quadratic function $q(x)=ax^2-bx+1$, where $a,b>0$, $q(x)$ reaches minimum at $x^*=b/2a$. Moreover, $q(0)=q(2x^*)$; and for all $x\in(0,2x^*)$, $q(x)\in[q(x^*),q(0))$. In our case, $a=\gamma$ and $b=2\beta$. Therefore, the minimum occurs at $\eta^*=\beta/\gamma$, and $\eta^*>0$.
    Since $\rho(0)=1$, for all $\eta\in(0,2\eta^*)$, $\rho(\eta)\in[\rho(\eta^*),1)$. 
    
    The remaining task is to compute the minimum value $\rho(\eta^*)$. Upon substituting $\eta^*=\beta/\gamma$ into $\rho(\eta)=\gamma\eta^2-2\beta\eta+1$, we obtain that $\rho(\eta^*)=1-\beta^2/\gamma$.

    Recall the definition of $\beta$ and $\gamma$ in Equation (\ref{eq-beta}, \ref{eq-gamma}), we can first see that $\beta\leq n\mu$. Also, since $\alpha_h\geq 1$, $\gamma\geq nL^2(h+b)=n^2L^2$. Finally, by Lemma \ref{lem:conv-1}, $\mu\leq L$. These imply that $\beta^2\leq n^2\mu^2\leq n^2L^2 \leq \gamma$; and thus $\rho(\eta^*)=1-\beta^2/\gamma>0$.
    
    In conclusion, for all $\eta>0$ such that $\eta<2\eta^*$, $\rho\in[0,1)$. This proves the theorem.
\end{proof}

\section{Proof of Lemmas in Section \ref{s:converge}}
\label{app:converge}

\subsection{Proof of Lemma \ref{lem:eq-os-4}}

\begin{proof}
    Consider any $j\in\h$. First, observe that we have
    \begin{align}
        \label{eq-os-1}
        \EE(\norm{g_j^t}) &= \EE(\norm{g_j^t-\nabla Q(w^t)+\nabla Q(w^t)}) \notag\\
        &\leq \EE\norm{g_j^t-\nabla Q(w^t)}+\EE\norm{\nabla Q(w^t)}.
    \end{align}
    By Jensen's inequality, for any random variable $X$, $\EE(X)^2\leq \EE(X^2)$. In our case, upon substituting $X=\norm{g_j^t-\nabla Q(w^t)}$, we obtain that
    \begin{equation}
        \label{eq-os-2}
        \EE\norm{g_j^t-\nabla Q(w^t)} \leq \sqrt{\EE\norm{g_j^t-\nabla Q(w^t)}^2}. \notag 
    \end{equation}
    By Assumption \ref{as:bounded-variance}, this further implies that
    \begin{equation}
        \label{eq-os-3}
        \EE\norm{g_j^t-\nabla Q(w^t)} \leq \sqrt{\sigma^2\norm{\nabla Q(w^t)}^2} = \sigma\norm{\nabla Q(w^t)}. \notag
    \end{equation}
    Upon substituting this into Equation (\ref{eq-os-1}), we obtain that
    \begin{equation}
        \EE\norm{g_j^t} \leq (1+\sigma)\norm{\nabla Q(w^t)}. \notag
    \end{equation}   
    This completes the proof.
\end{proof}

\subsection{Proof of Lemma \ref{lem:order-stats}}
\label{app:Lemma-7}

\begin{proof}
    Gumbel \cite{Gumbel-order-stats1} and Hartley and David \cite{Hartley-David-OrderStats2} proved that given identical means and variances $(\mu,\sigma^2)$, the upper bound of the expectation of the largest random variable among $n$ independent random variables is $\mu+\frac{\sigma(n-1)}{\sqrt{2n-1}}$. In our model, $\{g_j^t:j\in\h\}$ can be viewed as a set of independent and identically distributed random vectors with expectation $\nabla Q(w^t)$. Recall that $\h$ is the set of fault-free workers. Therefore, $\{\norm{g_j^t}:j\in\h\}$ is also a set of iid random variables; and by our algorithm, the norm of the received gradients $\norm{\Tilde{g}_j^t}=\norm{g_j^t}$.
    
    Denote the mean and variance of each $\norm{g_j^t}$ as $(\epsilon,\delta^2)$, and denote $M$ as $M = \max_{j\in\h}\{\norm{g_j^t}\}$. Thus, \cite{Gumbel-order-stats1,Hartley-David-OrderStats2} imply that $\EE M\leq \epsilon+\frac{\delta(h-1)}{\sqrt{2h-1}}$. $\epsilon$ and $\delta$  are to be bounded next.
    
    Recall that Lemma \ref{lem:eq-os-4} gives an upper bound on $\EE\norm{g_j^t}$, i.e.,
    \begin{equation}
        \epsilon = \EE\norm{g_j^t} \leq (1+\sigma)\norm{\nabla Q(w^t)}. \notag
    \end{equation}
    Therefore, we only need to compute the upper bound of $\delta^2$. 
    Consider a random vector $X\in\RR^d$ and denote $\mu=\EE X$. It is proved that for any constant vector $v\in\RR^d$, $\EE\dotprod{v}{X}=\dotprod{v}{\EE X}$. Therefore,
    \begin{align}
        \label{eq:bound-EX^2}
        \EE\norm{X-\mu}^2 
        &= \EE\dotprod{X-\mu}{X-\mu} \notag\\
        &= \EE\dotprod{X}{X}-2\EE\dotprod{\mu}{X}+\EE\dotprod{
        \mu}{\mu} \notag\\
        &= \EE\norm{X}^2 - \norm{\mu}^2.
    \end{align}
    Perlman \cite{PERLMAN197452} proved the extended Jensen's Inequality to random vectors such that $\phi(\EE X)\leq \EE\phi(X)$ for any convex function $\phi$. Since $\norm{\cdot}$ is convex, this implies $\norm{\mu}\leq \EE\norm{X}$. Thus,
    \begin{align}
        \mathrm{Var}\norm{X} 
        &= \EE\norm{X}^2 - (\EE\norm{X})^2 \notag\\
        &\leq \EE\norm{X}^2 - \norm{\mu}^2 \notag\\
        &= \EE\norm{E-\mu}^2. \notag
    \end{align}
    Upon substituting $X$ with $g_j^t$, we obtain that
    \begin{align}
        \label{eq:delta-square}
        \delta^2 = \mathrm{Var}\norm{g_j^t}
        &\leq \EE\norm{g_j^t-\nabla Q(w^t)}^2 \notag\\
        &\leq \sigma^2\norm{\nabla Q(w^t)}^2. & \text{(Assumption \ref{as:bounded-variance})}
    \end{align}
    The upper bounds of $\epsilon$ and $\delta^2$ together induce the following bound on $\EE M$, which is
    \begin{align}
        \label{eq-os-8}
        \EE M
        &\leq \epsilon + \frac{h-1}{\sqrt{2h-1}}\delta
        \leq \left((1+\sigma)+\frac{h-1}{\sqrt{2h-1}}\sigma\right) \norm{\nabla Q(w^t)}. \notag
    \end{align}
    By definition of $k_x$ in Equation (\ref{eq-kn}), this can be simplified as
    \begin{equation}
        \label{eq-order-stat-abbrev}
        \EE M \leq (1+k_h\sigma)\norm{\nabla Q(w^t)}. \notag
    \end{equation}

    Finally, recall the definition of the CGC filter in Equation (\ref{eq-CGC}). Since $|\h|\geq n-f$, it is guaranteed that one fault-free received gradient will be clipped. This implies that the threshold of the CGC filter $\norm{\Tilde{g}_{i_{n-f}}^t}$ satisfies $\norm{\Tilde{g}_{i_{n-f}}^t}\leq M$. Since the CGC filter guarantees that $\norm{\hat{g}_j^t}\leq\norm{\Tilde{g}_{i_{n-f}}^t}$ for all $j$, together we have $\EE\norm{\hat{g}_j^t}\leq \EE\norm{\Tilde{g}_{i_{n-f}}^t}\leq \EE M$, which proves the lemma.
\end{proof}

\subsection{Proof of Lemma \ref{lem:order-stat2}}

\begin{proof}
    Papadatos \cite{Papdatos-max-variance-OS} proved that for $n$ iid random variables $X_1\leq X_2\leq \dotsm\leq X_n$ with finite variance $\sigma^2$, the maximum variance of $X_n$ is bounded above by $n\sigma^2$. In the same setup as Section \ref{app:Lemma-7} that $M=\max\{\norm{g_j^t}:j\in\h\}$, Papadatos' result together with Equation (\ref{eq:delta-square}) imply that
    \begin{equation}
        \mathrm{Var}M \leq h\mathrm{Var}\norm{g_j^t} \leq h\sigma^2\norm{\nabla Q(w^t)}^2. \notag
    \end{equation}
    Recall that for a random variable $Y$, $\mathrm{Var}Y=\EE Y^2-(\EE Y)^2$. Therefore,
    \begin{align}
        \EE M^2 &= \mathrm{Var}M+(\EE M)^2 
        \leq (h\sigma^2+(1+k_h\sigma)^2)\norm{\nabla Q(w^t)}^2. \notag
    \end{align}
    Recall the definition of $\alpha_x$ in Equation (\ref{eq-alpha}). This implies
    \[
    \EE M \leq \alpha_h \norm{\nabla Q(w^t)}^2.
    \]
    Finally, in Section \ref{app:Lemma-7} we proved $\norm{\hat{g}_j^t}\leq M$ for all $j$, so $\norm{\hat{g}_j^t}^2\leq M^2$ and $\EE\norm{\hat{g}_j^t}^2\leq \EE M^2$ for all $j$, which proves the lemma.
\end{proof}

\subsection{Part C in Proof of Theorem \ref{thm:convergent}}
\label{app:part-C}

\begin{proof}
\noindent\textbf{Part $C$}: By definition of $g^t$ and convexity of $\norm{\cdot}^2$,
    \begin{align}
        \norm{g^t}^2 = \left\lVert\sum_{j=1}^n\hat{g}_j^t\right\rVert^2
        &\leq n\sum_{j=1}^n\norm{\hat{g}_j^t}^2 & \text{($\norm{\cdot}^2$ is convex)} \notag\\
        &= n\sum_{j\in\h}\norm{\hat{g}_j^t}^2+n\sum_{j\in\B}\norm{\hat{g}_j^t}^2. \notag
    \end{align}
    By definition of $\hat{g}_j^t$, $\norm{\hat{g}_j^t}\leq\norm{g_j^t}$ for all $j\in\h$. Also, by Equation (\ref{eq:bound-EX^2}), for each fault-free worker $j\in\h$,
    \begin{align}
        \label{eq-C2}
        \EE\norm{\hat{g}_j^t} \leq \EE\norm{g_j^t}^2
        &\leq \EE\norm{g_j^t-\nabla Q(w^t)}^2+\norm{\nabla Q(w^t)}^2 \notag\\
        &\leq \sigma^2\norm{\nabla Q(w^t)}^2 + \norm{\nabla Q(w^t)}^2.
    \end{align}
    By Lemma \ref{lem:order-stat2}, for each Byzantine worker $j\in\B$,
    \begin{equation}
        \label{eq-C2.5}
        \EE\norm{\hat{g}_j^t}^2 \leq \alpha_h\norm{\nabla Q(w^t)}^2.
    \end{equation}
    Recall the $L$-Lipschitz assumption (Assumption \ref{as:lipschitz}) that $\norm{\nabla Q(w^t)}\leq L\norm{w^t-w^*}$.
    Hence, upon combining Equation (\ref{eq-C2}) and (\ref{eq-C2.5}), we obtain that
    \begin{align*}
        \EE\norm{g^t}^2 
        &\leq nh(1+\sigma^2)\norm{\nabla Q(w^t)}^2 + nb\alpha_h\norm{\nabla Q(w^t)}^2 \\
        &\leq nL^2\left(h(1+\sigma^2)+b\alpha_h\right)\norm{w^t-w^*}^2.
    \end{align*}
    Recall the definition of $\gamma$ in Equation (\ref{eq-gamma}). This proves part C of Theorem \ref{thm:convergent}.
\end{proof}

\section{Proof in Section \ref{s:communication}}
\label{app:communication}

\subsection{Proof of Lemma \ref{lem:ball}}

\begin{proof}
    First, by triangular inequality, for all $u,v\in B$,
    \begin{align}
        \label{eq-lemcomp-1}
        \norm{u-v} &\leq \norm{u-\nabla Q(w^t)} + \norm{v-\nabla Q(w^t)}.
    \end{align}
    Again by triangular inequality and the radius of ball $B$, for each $u\in B$,
    \begin{align}
        \norm{\nabla Q(w^t)}-\norm{u} \leq \norm{u-\nabla Q(w^t)} \leq \frac{r}{2+r}\norm{\nabla Q(w^t)}. \notag
    \end{align}
    This is equivalent to say
    \begin{equation}
        \left(1-\frac{r}{2+r}\right)\norm{\nabla Q(w^t)} \leq \norm{u}
        \iff \norm{\nabla Q(w^t)} \leq \frac{2+r}{2}\norm{u}.
        \notag
    \end{equation}
    Therefore, 
    \begin{equation}
        \norm{u-\nabla Q(w^t)} \leq \frac{r}{2+r}\norm{\nabla Q(w^t)} \leq \frac{r}{2}\norm{u}, \enspace \forall u\in B. \notag
    \end{equation}
    Finally, upon substituting this into Equation (\ref{eq-lemcomp-1}), we prove the lemma.
\end{proof}

\section{Proof of existence of M-P inverse}
\label{appendix-MP-inverse}

\begin{proof}
    First of all, recall that for a matrix $A\in\RR^{m\times n}$ where $m>n$, if $A$ has full column rank, then $A^TA$ is invertible. 

    Consider arbitrary $j,t$. First we can prove by induction that all gradients $g_i^t$ stored in $R_j$ are linearly independent. Base case is true when $|R_j|=1$. Now assume that when $|R_j|=k$, $R_j$ is linearly independent. Denote $A$ as the matrix generated by $R_j$, i.e., $A=[g]_{g\in R_j}$. Then $A^TA$ is invertible and M-P inverse of $A$ exists, i.e., $A^+$ in line \ref{line:linear-independent} of Algorithm \ref{alg:echo-CGC} is well-defined. 
    
    Now suppose that $j$ receives a gradient $g$ and stores it to $R_j$. Then $g$ must pass the condition in line \ref{line:linear-independent}, i.e., $AA^+g\neq g$. Suppose by contradiction that $g$ is linearly dependent of $R_j$, then there exists $x\in\RR^{k}$ such that $g=Ax$. Note that $A^+A=I$, the identity matrix. This implies that $AA^+g=AA^+(Ax)=Ax=g$, which is a contradiction. Hence, $g$ is independent of $R_j$.

    In the induction proof, we already showed that $A^+$ always exists because $R_j$ is always linearly independent. This proves that the M-P inverse always exists.
\end{proof}

\end{document}